\documentclass[12pt]{amsart}
\usepackage{amsmath,amscd,amsfonts,amssymb, color,bbm, bm,
braket}
\usepackage{latexsym, graphicx, pstricks,rotating, enumerate}
\usepackage[pdftex,bookmarks,colorlinks]{hyperref}
\definecolor{darkblue}{rgb}{0.0,0.0,0.3}
\hypersetup{colorlinks=true,citecolor=darkblue,
unicode=true,pdftitle=cosa.pdf,pdfauthor=Ritabrata
Sengupta,bookmarks=false, urlcolor=darkblue} 
\numberwithin{equation}{section}
\newtheorem{prop}{Proposition}[section]
\newtheorem{defin}{Definition}[section]
\newtheorem{rem}{Remark}[section]

\newtheorem{theorem}{Theorem}[section]

\newtheorem{corollary}{Corollary}[section]

\newcommand{\what}{\widehat}
\newcommand{\wtilde}{\widetilde}
\newcommand{\tr}{\mathrm{Tr\,}}
\newcommand{\dd}{\mathrm{d}}
\newcommand{\re}{\mathrm{Re\,}}
\newcommand{\im}{\mathrm{Im\,}}
\newcommand{\cov}{\mathrm{Cov}}

\begin{document}
\title[On the KWM spectrum of a weakly stationary quatum
process]{On the Kolmogorov
- Wiener - Masani spectrum of a multi-mode weakly stationary
  quantum process}
\author{K. R. Parthasarathy}
\address{Theoretical Statistics and Mathematics Unit, Indian
Statistical Institute, Delhi Centre, 7 S J S Sansanwal Marg,
New Delhi 110 016, India}
\email[K R Parthasarathy]{\href{mailto:krp@isid.ac.in}{krp@isid.ac.in}}
\author{Ritabrata Sengupta}
\address{Theoretical Statistics and Mathematics Unit, Indian
Statistical Institute, Delhi Centre, 7 S J S Sansanwal Marg,
New Delhi 110 016, India}
\email[Ritabrata Sengupta]{\href{mailto:rb@isid.ac.in}{rb@isid.ac.in}}
\begin{abstract}
 
\par We introduce the notion of a $k$-mode weakly stationary
quantum process $\bm{\varrho}$ based on the canonical
Schr\"odinger pairs of position and momentum observables in
copies of $L^2(\mathbb{R}^k)$, indexed by an additive
abelian group $D$ of countable cardinality. Such observables
admit an autocovariance map $\wtilde{K}$ from $D$
into the space of real $2k \times 2k$ matrices. The map
$\wtilde{K}$ on the discrete group $D$ admits a
spectral representation as the Fourier transform of a
$2k \times 2k$ complex Hermitain matrix-valued totally
finite measure $\Phi$ on the compact character group
$\what{D}$, called the Kolmogorov-Wiener-Masani (KWM)
\emph{spectrum} of the process $\bm{\varrho}$. Necessary and
sufficient conditions on a $2k \times 2k$ complex Hermitian
matrix-valued measure $\Phi$ on $\what{D}$ to be the KWM
spectrum of a process $\bm{\varrho}$ are obtained. This
enables the construction of examples. Our theorem reveals
the dramatic influence of the uncertainty relations among
the position and momentum observables on the KWM spectrum of
the process  $\bm{\varrho}$. In particular, KWM spectrum
cannot admit a gap of positive Haar measure in $\what{D}$. 

\par The relationship between the number of photons in a
particular mode at any site of the process and its KWM
spectrum needs further investigation. 


\smallskip

\end{abstract}
\thanks{RS acknowledges financial
support from the National Board for Higher Mathematics,
Govt. of India.} 
\maketitle
\section{Introduction}

\par In his celebrated little book ``Osnovnye ponyatiya
teorii veroyatnostei" \cite{kol}, A. N. Kolmogorov
introduced the notion of a stochastic process as a
consistent family of finite dimensional probability
distributions in $\mathbb{R}^n,\, n=1,2,\cdots$. In the same
spirit a quantum process can be described as a consistent
family of density operators or, equivalently, states in
tensor products $\mathcal{H}_1 \otimes \cdots \otimes
\mathcal{H}_n$ of Hilbert spaces with $n=1,2,\cdots$. One
can replace the `time set' $\{1,2,\cdots\}$ by an abstract
countable set $D$ with the discrete topology and a family
$\{\mathcal{H}_a : a\in D\}$ of Hilbert spaces. Then a
quantum process yields a density operator $\rho_{a_1, a_2,
\cdots, a_n}$ in $\mathcal{H}_{a_1} \otimes \cdots
\otimes\mathcal{H}_{a_n}$ for every finite sequence $(a_1,
\cdots, a_n)$ with distinct elements from $D$. All these
density operators will obey natural consistency conditions.
For example, the relative trace of $\rho_{a_1, a_2, \cdots,
a_n}$ over $\mathcal{H}_{a_n}$ is $\rho_{a_1, a_2, \cdots,
a_{n-1}}$. If $(b_1,\cdots, b_n)$ is a permutation of $a_1,
\cdots, a_n$ then $\rho_{b_1, \cdots, b_n} = U
\rho_{a_1,\cdots, a_n}  U^{-1}$ where $U$ is the
corresponding Hilbert space isomorphism from
$\mathcal{H}_{a_1} \otimes \cdots \otimes \mathcal{H}_{a_n}$
onto  $\mathcal{H}_{b_1} \otimes \cdots \otimes
\mathcal{H}_{b_n}$ induced by the permutation. We denote the
quantum process over $D$ by 
\begin{equation}\label{13}
\bm{\varrho}=\left\{ \left(\mathcal{H}_{a_1, \cdots, a_n}, \rho_{a_1,
\cdots, a_n} \right): (a_1, \cdots, a_n) \in
\mathcal{S}_D\right\}
\end{equation}
where $\mathcal{S}_D$ denotes the set of all finite length
sequences of distinct elements from the countable set $D$. 

\par In this paper we are interested in the special case
where $\mathcal{H}_a = L^2(\mathbb{R}^k)$ for all $a$ in
$D$, $k$ being a fixed positive integer, called the
\emph{number of modes} of the process. Each $\mathcal{H}_a$
admits Schr\"odinger canonical pairs $q_{ar},\,p_{ar},\,
r=1,2,\cdots, k$ of position and momentum observables
obeying the Heisenberg canonical commutation relations (CCR).
We can look upon $q_{ar},\,p_{ar},\, r=1,2,\cdots, k$ as
observables in $\mathcal{H}_{a_1} \otimes  \cdots \otimes
\mathcal{H}_{a_n}$ whenever the sequence $(a_1, a_2, \cdots,
a_n)$ from $\mathcal{S}_D$ contains the element $a$ and
denote such ampliated observables by the same respective
symbols. With such a convention one obtains the algebra of
all polynomials of all $q_{ar},\,p_{ar},\, r=1,2,\cdots, k,
\, a \in D$ . Using the finite-partite states
$\rho_{a_1,\cdots, a_n},\, (a_1, \cdots, a_n) \in
\mathcal{S}_D$ one can compute the expectations of the
polynomials whenever they exist. Write
\[\left( X_{a\,1},  X_{a\,2}, \cdots, X_{a\,(2k-1)}, X_{a\,
2k} \right) = \left(q_{a\,1}, p_{a\, 1}, \cdots, q_{a\,k},
p_{a\,k} \right)\]
and define the covariances 
\begin{equation}\label{1.1}
\kappa_{r\,s} (a,b) = \left\langle\frac{1}{2}\left( X_{a\,r} X_{b\,
s} + X_{b\, s} X_{a\,r}\right)\right\rangle -
\braket{X_{a\,r}}\braket{X_{b\, s}}
\end{equation}
where $\braket{,}$ denotes expectation. To compute these
quantities we need a knowledge of only the `bipartite'
states $\rho_{a,b},\, (a,b)\in \mathcal{S}_D$. Thus we
obtain a $2k \times 2k$ real matrix-valued covariance kernel
$\mathcal{K} =[[K(a,b)]]$ defined by 
\begin{equation}\label{1.2}
K(a,b) = [[ \kappa_{r\,s}(a,b)]], \quad r,s, \in
\{1,2,\cdots,2k\},
\end{equation}
for $a,\, b\in D$.

\par Suppose $D$ is a countable discrete additive abelian
group with addition operation $+$ and null element $0$. Let
the covariance kernel $\mathcal{K}$ of a $k$-mode quantum
process over $D$ be translation invariant in the sense that 
\begin{equation}\label{1.3}
K(a+x, b+x) =K(a,b) \quad \forall a,b,x \in D.
\end{equation}
Then we say that the quantum process is second order
\emph{weakly stationary}, or, simply, weakly stationary. For
such a process there exists a map $\widetilde{K}$
from $D$ into the space of $2k \times 2k$ real matrices such
that
\begin{equation}\label{1.4}
K(a,b) =\widetilde{K}(b-a),\quad a,\,b
\in D.
\end{equation}
The map $\widetilde{K}$ is called the
\emph{autocovariance map} of the weakly stationary quantum
process. 

\par Owing to the matrix-positivity properties enjoyed by
covariances between observables the autocovariance map
$\widetilde{K}$ satisfies the matrix inequalities 
\begin{equation}\label{1.5}
\sum_{i,\,j} \alpha_i \alpha_j \widetilde{K}(a_j
-a_i) \ge 0
\end{equation}
for any $a_1, a_2, \cdots, a_n \in D$ and real scalars
$\alpha_1, \alpha_2, \cdots, \alpha_n,\, n=1,2,\cdots.$
Thanks to Bochner's  theorem for locally compact
abelian groups there exists a complex-Hermitian and positive
$2k \times 2k$ matrix-valued measure $\Phi$ on the compact
dual character group $\what{D}$ of $D$ such that 
\begin{equation}\label{1.6}
\widetilde{K}(a) =\int_{\what{D}} \chi(a) \,
\Phi(\dd\chi), \quad \text{for all $a$ in $D$}.
\end{equation}
The matrix-valued measure $\Phi$ satisfies the conjugate
symmetry property 
\begin{equation}\label{1.7}
\Phi(S^{-1}) = \overline{\Phi(S)}
\end{equation}
for any Borel set $S \subset \what{D}$. Furthermore, the
Heisenberg uncertainty relations prevailing among  the
various position and momentum observables of the quantum
process reveal their dramatic influence on the measure
$\Phi$ through the matrix inequalities 
\begin{equation}\label{1.8}
\Phi(S) + \frac{\imath}{2} \lambda(S) J_{2k} \ge 0
\end{equation}
for all Borel sets $S \subset \what{D}$, where $\lambda$  is
the normalised Haar measure of the compact group $\what{D}$
and $J_{2k}$ is the fundamental symplectic matrix given by 
\begin{equation}\label{1.9}
J_{2k} = \bigoplus_{\text{$k$-copies}} \begin{bmatrix} 0 & 1
\\ -1 & 0 \end{bmatrix}
\end{equation}
which is a diagonal block matrix with each diagonal block
equal to $J_2$.  The inequality (\ref{1.8}) implies, in
particular, that whenever $\Phi(S) =0, \, \lambda(S)$ is
also zero.

\par Borrowing from the extensive theory of linear least
square prediction of real valued weakly stationary processes
pioneered by A. N. Kolmogorov \cite{kol1,
kol2} and N. Wiener \cite{wiener}, and
multivariate weakly stationary processes by N. Wiener and P.
Masani \cite{WM1, WM2} we call (\ref{1.6}) the spectral
representation of $\widetilde{\mathcal{K}}$ in $\what{D}$
and the matrix-valued positive measure $\Phi$ the \emph{Kolmogorov
- Wiener - Masani spectrum} (or \emph{KWM spectrum}) of the
  autocovariance map $\widetilde{\mathcal{K}}$ of the underlying
quantum process. 

\par Inequality (\ref{1.8}) implies that whenever
$\Phi(S)=0$ for some Borel set $S \subset \what{D}$, then
$\lambda(S)=0$. In other words, the KWM spectrum does not
admit a `Haar gap'. 

\par Conversely, given a complex Hermitian positive $2k
\times 2k$ matrix-valued measure $\Phi$ on the Borel
$\sigma$-algebra of $\what{D}$ satisfying the conjugate
symmetry condition (\ref{1.7}), the spectral uncertainty
relations (\ref{1.8}), and the condition
$\Phi(\what{D}) =M$,
there exists a weakly stationary $k$-mode quantum process
over $D$ with KWM spectrum $\Phi$. Indeed, such a process
can be realized as a mean zero quantum Gaussian process in
the sense that all its finite-partite states
$\rho_{a_1,\cdots,a_n},\, (a_1, \cdots, a_n)  \in
\mathcal{S}_D$ are mean zero Gaussian states.

\par The spectral representation of the autocovariance
function and its converse enable us to construct interesting
examples of weakly stationary quantum processes. 

\section{Quantum processes}\label{sec:2}

\par A quantum system in its most elementary form is
determined by a pair $(\mathcal{H}, \rho)$ where
$\mathcal{H}$ is a complex separable Hilbert space and
$\rho$ is a density operator in $\mathcal{H}$, i.e., a
positive operator with unit trace. The operator $\rho$ is
called the state of the system. We shall deal with several
quantum systems and assume that all the Hilbert spaces in
this paper are complex and separable. Scalar products in
Hilbert spaces will be expressed in the Dirac notation and
adjoints of operators as well as matrices will be indicated
by the symbol $\dag$. By a positive operator $X$ in a
Hilbert space $\mathcal{H}$ we mean that $\braket{u|X|u} \ge
0$ for all $u \in \mathcal{H}$. By a positive $n \times n$
matrix we mean an $n \times n$ Hermitian matrix which is
positive semidefinite.

\par If $\mathcal{H} =  \mathcal{H}_1 \otimes \mathcal{H}_2
\otimes \cdots \otimes \mathcal{H}_n$ is the tensor product
of Hilbert spaces $\mathcal{H}_i, \, 1 \le i \le n, \, \rho$
is a state in $\mathcal{H}$ and $F \subset \{1,2,\cdots,n\}$
is the subset $\{i_1< i_2 < \cdots < i_k\}$ then we write
$\mathcal{H}_F = \mathcal{H}_{i_1} \otimes \cdots \otimes
\mathcal{H}_{i_k}$. One obtains a state $\rho_F$ in
$\mathcal{H}_F$ by taking the relative trace of $\rho$ 
successively in $\mathcal{H}_i,\, i\notin F$ in some order.
The resulting state $\rho_F$ is independent of the order
in which the traces are taken. The system $(\mathcal{H}_F,
\rho_F)$ is called the $F$-marginal of $(\mathcal{H},\rho)$.

\par In the Hilbert space of any quantum system a bounded or
unbounded self-adjoint operator $X$ is called an
\emph{observable} of the system. Suppose
$\mathcal{F}_\mathbb{R}$ is the Borel $\sigma$-algebra of
$\mathbb{R}$ and $P^X(\cdot)$ is the spectral measure of $X$
on $\mathcal{F}_\mathbb{R}$. Then the quantity $\tr \rho
P^X(E),\, E \in \mathcal{F}_\mathbb{R}$ is interpreted as
the probability that the observable takes a value in $E$ in
the state $\rho$. Thus $\tr \rho P^X(\cdot)$ is the
distribution of $X$ in the state $\rho$. Such an
interpretation enables the computation of all moments of
$X$. Indeed, the $n$-th moment of $X$, if it exists, is
denoted by $\braket{X^n}$ and is given by 
\[\braket{X^n} = \tr X^n \rho.\]
If $X$, $Y$ are two observables such that $XY+YX$ is also an
observable then the covariance between $X$ and $Y$ in the
state $\rho$ is denoted by $\cov(X,Y)$ and is defined as
\[\cov(X,Y) = \braket{\frac{1}{2}(XY+YX)} -
\braket{X}\braket{Y}.\]
The quantity $\cov(X,X)$ is called the variance of $X$.
If $X_1,\, X_2,\,\cdots,\, X_n$ are observables with
well-defined covariance between $X_i$ and $X_j$ for all
$i,\,j$  then the $n \times n$ positive matrix 
\[\Sigma_n = \Sigma_n (X_1, \cdots, X_n) =
[[\cov(X_i,X_j)]]\]
is called the \emph{covariance matrix} of the observables
$(X_1,X_2,\cdots, X_n)$ in the state $\rho$. 

\par Consider a composite quantum system
$(\mathcal{H},\rho)$ where $\mathcal{H} = \mathcal{H}_1
\otimes \mathcal{H}_2 \otimes \cdots \otimes \mathcal{H}_n$.
If the set $\{1,2,\cdots,n\} = E \cup F$ with $E\cap F
=\emptyset, \, E\neq \emptyset,\, F \neq \emptyset$ then
$\mathcal{H}$ can be viewed as the tensor product 
\[ \mathcal{H} =\mathcal{H}_E \otimes \mathcal{H}_F\]
and an observable in $\mathcal{H}_E$ can be looked upon as
the observable $X_E \otimes I_F$ in $\mathcal{H}$ with $I_F$
being the identity operator in $\mathcal{H}_F$. We call $X_E
\otimes I_F$ the \emph{ampliation} of $X_E$ in $\mathcal{H}$
	and denote it by the same symbol $X_E$. If $\rho_E$ is the
$E$-marginal of $\rho$ in $\mathcal{H}_E$ then 
\[\braket{X_E} = \tr X_E \rho_E = \tr X_E \rho = \braket{X_E
\otimes I_F}.\]

\par We now introduce the notion of a quantum process over a
countable index set $D$. Let $\{\mathcal{H}_a: a \in D\}$ be
a family of Hilbert spaces. Denote by $\mathcal{S}_D$ the
set of all finite sequences of distinct elements from $D$. 
Suppose $\rho_{a_1, a_2, \cdots ,a_n}$ is a density operator in
$\mathcal{H}_{a_1, a_2,\cdots, a_n}$ as in 
(\ref{13}) for each $(a_1, a_2,
\cdots, a_n)$ in $\mathcal{S}_D$ satisfying the following
properties: 
\begin{enumerate}
\item If $\{a_1, a_2 \cdots, a_n\} = \{b_1, b_2, \cdots,
b_n\}$ as sets and $\pi$ is a permutation of
$\{1,2,\cdots,n\}$ such that $a_{\pi(j)}=b_j,\, \forall j$
and 
\[U_\pi:\mathcal{H}_{a_1, a_2,\cdots, a_n}  \rightarrow
\mathcal{H}_{b_1, b_2,\cdots, b_n}\]
is the natural Hilbert space isomorphism induced by $\pi$
then 
\[\rho_{b_1, b_2, \cdots, b_n} = U_\pi \rho_{a_1 ,a_2
,\cdots,
a_n} U_\pi^{-1}.\]

\item The $\{a_{n+1}\}$-marginal of $\rho_{a_1, a_2, \cdots a_{n+1}}$
is equal to $\rho_{a_1, a_2, \cdots, a_n}$ for all $(a_1, a_2,
\cdots, a_{n+1}) \in \mathcal{S}_D,\, n=1,2,\cdots$.
\end{enumerate}
Then we say that $\{\rho_{a_1, a_2, \cdots, a_n}, (a_1, \cdots,a_n)\in
\mathcal{S}_D\}$ is a consistent family of states. The
family $\bm{\varrho}= \{(\mathcal{H}_{a_1,\cdots, a_n}, \rho_{a_1,\cdots,
a_n}), (a_1,\cdots ,a_n) \in \mathcal{S}_D\}$
of finite-partite quantum systems is called a \emph{quantum process} over $D$.

\par One obtains interesting examples of discrete `time'
quantum processes with $D = \mathbb{Z},\, \mathbb{Z}^d$ or a
general discrete abelian group. When $D$ is $\mathbb{Z}$,
the element $a$ in $\mathbb{Z}$ can be interpreted as time.
In general, $a$ in $D$ is interpreted as site. 

\par Suppose $D=\{0,1,2,\cdots\}$ and $\mathcal{H}_{n]} =
\mathcal{H}_0 \otimes \mathcal{H}_1 \otimes \cdots \otimes
\mathcal{H}_n$. Let $\rho_{n]}$ be a density operator in
$\mathcal{H}_{n]}$ such that $\rho_{n-1]}$ is the marginal
in $\mathcal{H}_{n-1]}$ obtained by tracing out $\rho_{n]}$
over $\mathcal{H}_n$ for each $n$. Then
$\{(\mathcal{H}_{n]}, \rho_{n]}): n=0,1,2,\cdots\}$  yields a
quantum process. Denote by $\mathcal{B}_{n]}$ the C*
algebra of all bounded operators in $\mathcal{H}_{n]}$. Then there
is a natural C* embedding $\phi_n:\mathcal{B}_{n]}
\rightarrow \mathcal{B}_{n+1]}$ with the property 
\[\phi_n(X) = X \otimes I,\quad X \in \mathcal{B}_{n]},\]
where $I$ is the identity operator in $\mathcal{H}_{n+1}$. This
enables the construction of an inductive limit C* algebra
$\mathcal{B}_\infty$ with a C* embedding
$\pi_n:\mathcal{B}_{n]} \rightarrow \mathcal{B}_\infty$ such
that the sequence $\{\pi_n(\mathcal{B}_{n]})\}$ is
increasing in $n$ and $\bigcup_n \pi_n(\mathcal{B}_{n]})$ is
dense in $\mathcal{B}_\infty$. This yields a normalized
positive linear functional $\omega$ in $\mathcal{B}_\infty$
such that
\[\omega(\pi_n(X)) = \rho_{n]} (X), \quad X \in
\mathcal{B}_{n]}, \, n=0,1,2,,\cdots.\]
In other words $(\mathcal{B}_\infty,\omega)$ is a C*
probability space which may be considered as the analogue of
Kolmogorov's measure space constructed from a consistent
family of finite dimensional probability distributions.
However, there is no limiting Hilbert space in general with a density
operator. A similar construction of a C* probability space
is possible for a quantum process over any countable index set
$D$.

\begin{defin}
Suppose $D$ is a countable abelian group with addition
operation $+$, $\mathcal{H}_a = \mathcal{H}$ for all $a \in
D$, and $\bm{\varrho}$
is a quantum process over $D$. Then it is
said to be \emph{strictly stationary} or \emph{translation
invariant} if 
\[ \rho_{a_1+x, a_2+x, \cdots , a_n+x } = \rho_{a_1, a_2,
\cdots, a_n} \quad \forall x \in D, \, (a_1,\cdots,a_n) \in
\mathcal{S}_D. \]
\end{defin}

\par Let $\{\rho_{a_1,\cdots, a_n} \},\,
\{\sigma_{a_1,\cdots, a_n}\}, \, (a_1, \cdots, a_n) \in
\mathcal{S}_D$ be a pair of consistent families of
finite-partite states in $\{\mathcal{H}_{a_1,\cdots,
a_n}\}$. Then, for any $0< p<1$
\[\tau_{a_1,\cdots, a_n} = p \rho_{a_1,\cdots, a_n} +(1-p)
\sigma_{a_1,\cdots, a_n}, \quad (a_1,\cdots, a_n)\in
\mathcal{S}_D\]
yields a consistent family of finite-partite states.

\par Suppose, $a \mapsto U_a, \, a\in D$ is any map where
$U_a$ is a unitary operator in $\mathcal{H}_a$ for every
$a$. Then 
\[\rho_{a_1,\cdots, a_n}' = \left(U_{a_1} \otimes \cdots
\otimes U_{a_n} \right) \rho_{a_1,\cdots, a_n}
\left(U_{a_1}^\dag \otimes \cdots \otimes U_{a_n}^\dag \right), \quad (a_1,
\cdots, a_n) \in \mathcal{S}_D\]
also yields a consistent family of states. Indeed, this is a
consequence of the following proposition.
\begin{prop}
Let $\mathcal{H},\, \mathcal{K}$ be Hilbert spaces, $\rho$
a state in $\mathcal{H} \otimes \mathcal{K}$, and $U,\,
V$ be unitary operators in $\mathcal{H}$ and  $\mathcal{K}$
respectively. Then 
\[\tr_\mathcal{K} (U \otimes V) \rho (U \otimes V)^\dag = U
\left( \tr_\mathcal{K}\, \rho\right) U^\dag,\]
where $\tr_\mathcal{K}$ is relative trace over
$\mathcal{K}$.
\end{prop}
\begin{proof}
This is immediate from the fact that relative trace over
$\mathcal{K}$ can be computed by using any orthonormal basis
$\{\bm{e}_j\}$ in $\mathcal{K}$, and if $\{\bm{e}_j\}$ is
one such basis so is $\{V^\dag \bm{e}_j\}$.
\end{proof}

\par Combining the two elementary remarks above we can
construct new quantum processes over $D$ from a given
quantum process $\left\{ \left( \mathcal{H}_{a_1, \cdots,
a_n}, \rho_{a_1, \cdots, a_n} \right), (a_1, \cdots, a_n)
\in \mathcal{S}_D\right\}$ as follows: Start with a
probability space $( \Omega, \mathcal{F}, \mathcal{P})$ and
a random process $\{ U_a(\omega): a \in D\}$ where
$U_a(\omega)$ is a unitary operator in $\mathcal{H}_a$ for every
$a$. Define
\begin{equation}\label{9.1}
\rho_{a_1, \cdots, a_n}' = \int_\Omega P(\dd\omega)
\left(U_{a_1} \otimes \cdots \otimes U_{a_n} \right)
\rho_{a_1,\cdots, a_n} \left(U_{a_1} \otimes \cdots \otimes
U_{a_n} \right)^\dag.
\end{equation}
Then $\{\rho_{a_1,\cdots, a_n}' ,(a_1, \cdots, a_n) \in
\mathcal{S}_D\}$ is also a consistent family of
finite-partite states. 

\begin{rem}
When $D$ is a countable additive abelian group and
$\mathcal{H}_a = \mathcal{H}$ for all $a \in D$,
$\bm{\varrho}$
is a strictly stationary quantum process and the random
process $\{U_a(\omega): a \in D\}$ is also strictly
stationary, then the quantum process 
$\bm{\varrho}'$ determined by equation (\ref{9.1}) is also strictly stationary. 
\end{rem}
\section{Multi-mode processes and their covariance kernels} \label{sec:3}

\par We now pass on to the definition of a \emph{$k$-mode
quantum process} over a countable index set $D$. Let $\mathcal{H}_a =
L^2(\mathbb{R}^k)$ for each $a \in D$, where $k$ is a fixed
positive integer called the number of modes. We view
$\mathcal{H}_a$ as the $a$-th copy of $L^2(\mathbb{R}^k)$
and introduce the canonical Schr\"odinger pairs of position
and momentum observables $q_{aj},\, p_{aj},\, 1 \le j \le
k$ given by
\begin{eqnarray*}
\left( q_{aj} f\right) (\bm{x}) &=& x_j f(\bm{x}),\\
\left( p_{aj} f\right) (\bm{x}) &=& \frac{1}{\imath}
\frac{\partial }{\partial x_j} f(\bm{x})
\end{eqnarray*}
on their respective maximal domains in $L^2(\mathbb{R}^k)$,
$\bm{x}$ denoting $(x_1, x_2, \cdots, x_k) \in
\mathbb{R}^k$. We arrange these $2k$ observables as 
\[\left( X_{a1}, X_{a2}, \cdots, X_{a\,{2k-1}}, X_{a{2k}} \right) =
\left( q_{a1}, p_{a1},\cdots, q_{ak}, p_{ak} \right). \]
Let now $\bm{\varrho}$ 
be a quantum process over $D$. Then
$X_{a\,r}$ can be viewed as an ampliated observable in
$\mathcal{H}_{a_1, a_2, \cdots, a_n}$ whenever the element
$a$ occurs in the sequence $a_1, a_2, \cdots, a_n$. We
assume that all observables which are closures of polynomials
of degree not exceeding $2$ in 
$\{X_{a\,r},\, a \in D,\, r =1,2,\cdots,2k\}$ have finite
expectations under the process so that $X_{a\,r}$ and
$X_{b\,s}$ have a  well-defined covariance for any $a,b\in
D$, $r,s\in \{1,2,\cdots,2k\}$. We write for any $a,b \in D$
\[ \kappa_{r,s}(a,b) = \cov \left( X_{a\,r}, X_{b\,s}
\right),\quad r,\,s \in \{1,2,\cdots,2k\},\]
where the covariances can be evaluated in any state
$\rho_{a_1, a_2, \cdots, a_n}$ when both $a$ and $b$ occur in
$(a_1,\cdots,a_n) \in \mathcal{S}_D$.
Indeed, this follows from consistency of the
states occurring in the quantum process. We call
$\mathcal{K} = [[K(a,b)]],\, a,b \in D$ the \emph{covariance
kernel} of the $k$-mode quantum process $\bm{\varrho}$.

\par If $D$ is an additive abelian group, $\braket{X_{aj}}
=0$ for all $a\in D,\, 1 \le j \le 2k$ and $K(a,b) = K(a+c,
b+c)$ for all $a,\, b,\, c\in D$ we then say that
$\bm{\varrho}$
 is a mean zero \emph{second order weakly stationary} or
simply \emph{weakly stationary} $k$-mode quantum process. In
such a case, there exists a map $\wtilde{K}$ from $D$ into
the space of $2k \times 2k$ real matrices, such that
$K(a,b) = \wtilde{K}(b-a)$. This map
$\wtilde{K}$
is called the \emph{autocovariance map} of the weakly
stationary process. 

\begin{theorem}\label{th:1.1}
Let $K(a,b)=[[\kappa_{r\,s}(a,b)]],\, a,b\in D$ be a family of $2k \times 2k$ real
matrices satisfying the following conditions:
\[\kappa_{r\,s}(a,b) =\kappa_{s\,r}(b,a),\quad r,\,s\in
\{1,2,\cdots,2k\},\quad a,b \in D.\]
Then there exists a $k$-mode quantum process $\bm{\varrho}$
with covariance kernel $K(\cdot,\cdot)$ if and
only if for any sequence $(a_1, \cdots,a_n)\in
\mathcal{S}_D$ the block matrix $[[ K(a_i,a_j) ]]$ satisfies
the matrix inequality 
\begin{equation}\label{ce} 
[[ K(a_i, a_j) ]] + \frac{\imath}{2} J_{2kn} \ge 0.
\end{equation}
\end{theorem} 

\begin{proof}
Since $\rho_{a_1,\cdots, a_n}$ is a $kn$-mode
state and $[[K(a_i,a_j)]]$ is the covariance matrix of the
position-momentum observables  $\left( X_{a_1 \, 1},
\cdots , X_{a_1 \, 2k} , X_{a_2 \, 1}, \cdots , X_{a_2
\, 2k}, \cdots , X_{a_n \, 1}, \cdots , X_{a_n \, 2k}
\right)$ in $L^2(\mathbb{R}^{kn})$, 
necessity is immediate from the uncertainty relation
fulfilled by such a covariance matrix \cite{holevo1, arvindsurvey,
MR2662722}.  To prove the converse
define $\rho_{a_1,\cdots , a_n}$ to be the mean zero
$kn$-mode Gaussian state with covariance matrix $[[ K(a_i,
a_j)]]$. Then $\{\rho_{a_1,\cdots, a_n}\}$ is a consistent
family of Gaussian states constituting the required quantum
process.  
\end{proof}
\begin{defin}
\par A kernel 
\[\mathcal{K}=[[K(a,b)]],\, a,b \in D\]
where $K(a,b)$ are real $2k \times 2k$ matrices satisfying
the conditions 
\begin{enumerate}
\item $K(a,b)^T = K(b,a)$.
\item $[[K(a_i, a_j)]] \ge 0$ for all $(a_1, a_2, \cdots,
a_n) \in \mathcal{S}_D$
\end{enumerate}
is called a $k$-mode \emph{classical covariance} kernel. 

\par If, in addition, the inequality (\ref{ce}) is fulfilled,
then it is called a ($k$- mode) \emph{quantum covariance
kernel}.
\end{defin}
\begin{corollary}
If $\mathcal{K}$ is a $k$-mode quantum covariance kernel and
$\mathcal{C}$ is a $k$-mode classical covariance kernel then
$\mathcal{K}+ \mathcal{C}$ is a $k$-mode quantum covariance
kernel.
\end{corollary}
\begin{proof} Immediate. \end{proof}

\par Let $\bm{\varrho}$ be a $k$-mode
quantum process over $D$ with quantum covariance kernel
$\mathcal{K} =[[ K(a,b)]]$, $a,b \in D$. Suppose
$\mathcal{C} = [[C(a,b)]]$, $a,b \in D$  is the covariance kernel of a real
$2k$-variate classical stochastic process so that the matrix
inequalities 
\[\sum_{i,j} \alpha_i \alpha_j C(a_i, a_j) \ge 0\]
for all real scalars $\alpha_1, \cdots, \alpha_n$, elements
$a_1, \cdots, a_n \in D$. Then the sum 
\[\mathcal{K} + \mathcal{C} = [[ K(a,b) + C(a,b) ]]\]
is the covariance kernel of a $k$-mode quantum process
$\bm{\sigma}$. We shall now realise such a process
$\bm{\sigma}$ by an explicit construction which is an
interaction between the quantum process $\bm{\varrho}$ and a
family of unitary conjugations mediated by a classical
process with covariance kernel $\mathcal{C}$.

\par To this end we start with the $1$-mode Hilbert space
$L^2(\mathbb{R})$, its Schr\"odinger position-momentum pair
$q,\,p$, the associated annihilation-creation pair $a,
\,a^\dag$ given by $a=2^{-\frac{1}{2}}(q + \imath p), \,
a^\dag=2^{-\frac{1}{2}}(q - \imath p)$  and the unitary Weyl
(displacement) operators $W(z) = \exp(za^\dag - \bar{z}
a),\, z \in \mathbb{C}$ satisfying the relations 
\[W(z) a W(z)^\dag = a -z , \quad z \in \mathbb{C}\]
with the convention that $z$ denotes the scalar as well as
the operator $z I$. This leads to the relations 
\begin{eqnarray}
W(2^{-\frac{1}{2}} z) q W(2^{-\frac{1}{2}} z)^\dag &=& q -x,
\label{f1} \\
W(2^{-\frac{1}{2}} z) p W(2^{-\frac{1}{2}} z)^\dag &=& p - y,
\label{f2}
\end{eqnarray} 
where $x = \re z, \, y = \im z$.

\par Now for $a \in D$, let 
\begin{eqnarray*}
\bm{z}_a &=& (z_{a\, 1}, z_{a\, 2}, \cdots , z_{a\, k})^T, \\
z_{a\, r} &=& x_{a\, r}+ \imath y_{a\, r},
\end{eqnarray*}
where $x_{a\,r} =\re z_{a\,r},\, y_{a\,r} =\im z_{a\,r}$.
Viewing $\mathcal{H}= L^2(\mathbb{R}^k)$ as $L^2(\mathbb{R})
\otimes L^2(\mathbb{R}) \otimes \cdots \otimes
L^2(\mathbb{R})$, $k$-fold, introduce the $k$-mode Weyl
operators
\[W(\bm{z}_a) = W(z_{a\, 1}) \otimes \cdots \otimes W(z_{a\,
k}).\]
Then the relations (\ref{f1}, \ref{f2}) yield the relations for
the operators $X_{a\,1},\cdots , X_{a\, 2k}$ defined above
as 
\begin{equation}\label{f3}
 W(2^{-\frac{1}{2}} \bm{z}_a) 
\begin{bmatrix}
X_{a\,1} \\ X_{a\,2} \\ \vdots \\ X_{a\, 2k} 
\end{bmatrix} 
W(2^{-\frac{1}{2}} \bm{z}_a)^\dag
=  \begin{bmatrix}
X_{a\,1} - \alpha_{a\, 1} \\ X_{a\,2} - \alpha_{a\,2}\\
\vdots \\ X_{a\, 2k} - \alpha_{a \, 2k} 
\end{bmatrix} 
\end{equation}
where 
\begin{equation}\label{f4}
\left( \alpha_{a\, 1}, \alpha_{a\, 2}, \cdots,
\alpha_{a\, 2k} \right) = \left( x_{a\, 1},y_{a\, 1}, x_{a\,
2},y_{a\, 2},\cdots, x_{a\, k},y_{a\, k}\right).
\end{equation}

\par Let $\left(\xi_{a\, 1},\eta_{a\, 1},
\xi_{a\,2},\eta_{a\, 2},\cdots, \xi_{a\, k},\eta_{a\,
k}\right)(\omega),\, \omega \in \Omega, \, a \in D$ be a
$2k$ real variate stochastic process defined on a
probability space $(\Omega, \mathcal{F}, P)$ with zero means
and covariance kernel $\mathcal{C} = [[ C(a,b) ]]$ where 
\begin{equation}\label{f4a}
C(a,b) = \mathbb{E} 
\begin{bmatrix}
\xi_{a\, 1}\\\eta_{a\, 1} \\ \xi_{a\,2} \\ \eta_{a\, 2} \\
\vdots \\ \xi_{a\, k} \\\eta_{a\,k} 
\end{bmatrix} \begin{bmatrix} \xi_{a\, 1},\eta_{a\, 1},
\xi_{a\,2},\eta_{a\, 2},\cdots, \xi_{a\, k},\eta_{a\,k}
\end{bmatrix}, \quad a,\, b \in D .
\end{equation}
Define the random unitary operators $U_a(\omega)$ in
$\mathcal{H}_a, \, a \in D$ by putting 
\begin{eqnarray}
\zeta_{a \, r} &=& \xi_{a\,r} + \imath \eta_{a\, r}, \quad r
=1,2,\cdots, k \label{f5} \\
U_a(\omega) &=& W(2^{-\frac{1}{2}}
\bm{\zeta}_a(\omega))^\dag \label{f6}
\end{eqnarray}
where $\bm{\zeta}_a=( \zeta_{a\,1} \cdots, \zeta_{a\,k}) \in
\mathbb{C}^k$. By following the remarks in \S \ref{sec:2}
equation (\ref{9.1}), define the $k$-mode quantum process
$\bm{\sigma}$ by 
\begin{equation}\label{f7}
\sigma_{a_1, \cdots, a_n} = \int_\Omega P(\dd\omega)
U_{a_1}(\omega) \otimes \cdots \otimes U_{a_n} (\omega)
\rho_{a_1,\cdots, a_n} U_{a_1}(\omega)^\dag \otimes \cdots \otimes
U_{a_n}(\omega)^\dag.
\end{equation}
Then we have the following theorem:
\begin{theorem}\label{th:f1}
The covariance kernel of the $\bm{\sigma}$ process
determined by the finite-partite states (\ref{f7}) is equal
to $\mathcal{K} + \mathcal{C}$.
\end{theorem}

\begin{proof}
Consider the observable $X_{a\,j}$. Its expectation under
the $\bm{\sigma}$ process is given by 
\begin{eqnarray*}
\tr X_{a\, r} \sigma_a &=& \int P(\dd\omega) \tr
U_a(\omega)^\dag X_{a\,r} U_a(\omega) \rho_a \\
&=& \int P(\dd\omega) \tr \left( X_{a\, r} - \gamma_{a\, r}
\right) \rho_a
\end{eqnarray*}
where 
\[\left(\gamma_{a\, 1}(\omega), \cdots, \gamma_{a\, 2k}(\omega)\right)
= \left( \xi_{a\, 1}(\omega),\eta_{a\, 1}(\omega),
\cdots, \xi_{a\, k}(\omega),\eta_{a\,k}(\omega) \right).\]
Since the classical $\bm{\gamma}$ process has mean $\bm{0}$
we have 
\begin{equation}\label{f8}
\braket{X_{a\,r}}_{\bm{\sigma}} =
\braket{X_{a\,r}}_{\bm{\rho}}.
\end{equation}
Going to second order moments 
\begin{eqnarray*}
\tr X_{a\, r} X_{b\, s} \,\sigma_{a\,b} &=&  \int
P(\dd\omega) \tr X_{a\, r} X_{b\, s}  U_a(\omega) \otimes
U_b(\omega) \rho_{a\,b}  U_a(\omega)^\dag \otimes
U_b(\omega)^\dag\\
&=& \int P(\dd\omega)
W(2^{-\frac{1}{2}}\bm{\zeta}_a) \otimes
W(2^{-\frac{1}{2}}\bm{\zeta}_b) X_{a\, r} X_{b\, s}
W(2^{-\frac{1}{2}}\bm{\zeta}_a)^\dag \otimes
W(2^{-\frac{1}{2}}\bm{\zeta}_b)^\dag  \rho_{a\,b}\\
&=& \int P(\dd\omega) \tr \left( X_{a\, r} - \gamma_{a\,
r}\right) \left( X_{b\, s} - \gamma_{b\,
s}\right)\rho_{a\,b}\\
&=& \int
P(\dd\omega)\left[\braket{X_{a\,r}X_{b\,s}}_{\bm{\varrho}} +
\gamma_{a\,r}(\omega)\gamma_{b\,s}(\omega) -
\gamma_{a\,r}(\omega) \braket{X_{b\,s}} -
\gamma_{b\,s}(\omega) \braket{X_{a\,r}}\right] \\
&=& \braket{X_{a\,r}X_{b\,s}}_{\bm{\varrho}} +
C_{r\,s}(a,b).
\end{eqnarray*}
Let $a \neq b$. Then
\begin{eqnarray*}
\cov_{\bm{\sigma}}(X_{a\,r}, X_{b\,s}) &=&
\braket{X_{a\,r}X_{b\,s}}_{\bm{\sigma}} -
\braket{X_{a\,r}}_{\bm{\sigma}}
\braket{X_{b\,s}}_{\bm{\sigma}} \\
&=& \braket{X_{a\,r}X_{b\,s}}_{\bm{\varrho}} + C_{r\,s}(a,b)
- \braket{X_{a\,r}}_{\bm{\varrho}}
\braket{X_{b\,s}}_{\bm{\varrho}} \\
&=& K_{r\,s}(a,b) +  C_{r\,s}(a,b).
\end{eqnarray*}
Let $a=b$. Then $C_{r\,s}(a,b)= C_{r\,s}(a,a) =
C_{s\,r}(a,a)$.
\begin{eqnarray*}
\cov_{\bm{\sigma}}(X_{a\,r}, X_{a\,s}) &=&
\left\langle \frac{1}{2}(X_{a\,r} X_{a\,s} + X_{a\,s}
X_{a\,r}) \right\rangle_{\bm{\sigma}} -
\braket{X_{a\,r}}_{\bm{\varrho}}
\braket{X_{a\,s}}_{\bm{\varrho}} \\
&=& \left\langle \frac{1}{2}(X_{a\,r} X_{a\,s} + X_{a\,s}
X_{a\,r}) \right\rangle_{\bm{\varrho}} + C_{r\,s}(a,a) - 
\braket{X_{a\,r}}_{\bm{\varrho}}
\braket{X_{a\,s}}_{\bm{\varrho}}\\
&=&  \cov_{\bm{\varrho}}(X_{a\,r}, X_{a\,s}) +
C_{r\,s}(a,a)\\
&=&  K_{r\,s}(a,a) +  C_{r\,s}(a,a).
\end{eqnarray*}
\end{proof} 
\begin{rem}
\par If $\bm{\varrho}$ is a weakly stationary quantum
process,
and $\left(\xi_{a\, 1},\eta_{a\, 1}, \cdots, \xi_{a\,
k},\eta_{a\,k}\right)$ is
a weakly stationary classical process with mean $\bm{0}$
such that
\begin{eqnarray*}
K(a,b) &=& \widetilde{K}(b-a), \quad a,b\in D,\\
C(a,b) &=& \widetilde{C}(b-a),
\end{eqnarray*}
then $\bm{\sigma}$ is also a weakly stationary quantum
process.
\end{rem}

\begin{rem}

\par If $\bm{\varrho}$ is a Gaussian process then so is
$\bm{\sigma}$. 
\end{rem}


\section{The KMW spectrum of a weakly stationary $k$-mode quantum
process}\label{sec:4}

\par Let $\bm{\varrho}$ be a weakly stationary $k$-mode
quantum process over a countable discrete additive abelian
group $D$, with autocovariance map
$\wtilde{K}$. Let $\what{D}$ be the compact dual
multiplicative group of all characters of $D$. Denote by
$\mathcal{F}$ the Borel $\sigma$-algebra on $\what{D}$. 

\par Define 
\begin{equation}\label{eq:2.3}
L(a) = \wtilde{K}(a) + \frac{\imath}{2} \mathbbm{1}_{\{0\}}(a) J_{2k},
\quad a \in D.
\end{equation}
Then Theorem \ref{th:1.1} yields the following proposition. 
\begin{prop}\label{prop:2.1}
A real $2k \times 2k$ matrix-valued map
$\wtilde{K}$ is the
autocovariance map of a second order weakly stationary
$k$-mode quantum process if and only if the associated
map $L$ defined by (\ref{eq:2.3}) satisfies the
following matrix inequalities
\[ [[ L(a_s - a_r) ]] \ge 0, \quad r,\, s\in \{1,2,\cdots,
n\}\]
for all $(a_1, a_2, \cdots, a_n) \in \mathcal{S}_D,\, n=1,2,
\cdots$.
\end{prop}
\begin{proof}
Immediate.
\end{proof}

\begin{theorem}\label{th:2.2}
A real $2k \times 2k$ matrix-valued map
$\wtilde{K}$ on $D$ is the
autocovariance map of a second order weakly stationary
$k$-mode quantum process on $D$ if and only if there exists
a $2k \times 2k$ Hermitian positive matrix-valued measure
$\Phi$ on $(\widehat{D},\mathcal{F})$ satisfying the following
conditions:
\begin{enumerate}
\item $\Phi(\what{D}) = \wtilde{K}(0)$,
\item $\wtilde{K}(a) = \int_{\what{D}} \chi(a) \Phi(\dd\chi)$,
\item $\Phi(S) + \frac{\imath}{2} \lambda(S) J_{2k} \ge 0,
\, \forall S \in \mathcal{F}$; where $\lambda$ is the
normalized Haar measure of the compact group $\what{D}$. In
particular, $\lambda$ is absolutely continuous with respect
to $\tr \Phi$.
\item $\Phi(S^{-1}) = \overline{\Phi(S)} = \Phi(S)^T,\,
\forall S \in \mathcal{F}$.
\end{enumerate}
\end{theorem}
\begin{proof}
Let $\wtilde{K}$ be the autocovariance map of a weakly
stationary $k$-mode process. Define $L$ by (\ref{eq:2.3}).
By Proposition \ref{prop:2.1}  the matrices $[[
L(a_s-a_r)]],\, r,s\in\{1,2,\cdots, n\},\, (a_1, a_2,
\cdots, a_n) \in \mathcal{S}_D$ are positive. Hence for any
vector $\bm{u} \in \mathbb{C}^{2k}$, the function 
\begin{equation}\label{eq:2.4}
\psi_{\bm{u}}(a) = \bm{u}^\dag L(a) \bm{u},\quad a \in D 
\end{equation}
is positive definite on the abelian group $D$ in the sense
of Bochner. By Bochner's theorem there exists a totally
finite measure $\nu_{\bm{u}}$ on $\mathcal{F}$ satisfying
the relations
\begin{eqnarray}
\psi_{\bm{u}}(a) &=& \int_{\what{D}} \chi(a) \,
\nu_{\bm{u}}(\dd\chi), \quad a \in D, \label{eq:2.5}\\
\psi_{\bm{u}}(0) &=& \bm{u}^\dag L(0) \bm{u} \nonumber \\
	&=& \bm{u}^\dag \left(\wtilde{K}(0) + \frac{\imath}{2}
J_{2k} \right) \bm{u}, \quad \bm{u} \in \mathbb{C}^{2k}.
\label{eq:2.6}
\end{eqnarray}
By (\ref{eq:2.4}) the left hand side of (\ref{eq:2.5}) is a
quadratic form in $\bm{u}$ for each fixed $a$ in $D$. By the
bijective correspondence between totally finite measures on
$\what{D}$ and their Fourier transforms on $D$ it follows
that there exists a $2k \times 2k$ Hermitian positive
matrix-valued measure $\Psi$ on $\mathcal{F}$ such that
\begin{eqnarray}
L(a) &=& \int_{\what{D}} \chi(a) \, \Psi(\dd\chi), \quad a
\in D \label{eq:2.7} \\
L(0) &=& \wtilde{K}(0) + \frac{\imath}{2} J_{2k} = \Psi(\what{D}) \ge
0. \label{eq:2.8}
\end{eqnarray}
Now define
\[ \phi_{\bm{u}}(a) = \bm{u}^\dag \wtilde{K}(a) \bm{u}, \quad a \in
D, \quad \bm{u} \in \mathbb{C}^{2k}.\]
By (\ref{eq:2.3}), $\wtilde{K}(a) = \re L(a)$ and
hence $[[\wtilde{K}(a_s -
a_r) ]] \ge 0$ for any $(a_1,a_2, \cdots, a_n) \in
\mathcal{S}_D$. In other words, $\phi_{\bm{u}}$ is also a
positive definite function on $D$ and by the same arguments
as employed for $L$ we have the relations
\begin{eqnarray}
\wtilde{K}(a) &=& \int_{\what{D}} \chi(a)\, \Phi(\dd\chi),
\label{eq:2.9} \\
\wtilde{K}(0) &=& \Phi(\what{D}), \label{eq:2.10}
\end{eqnarray}
where  $\Phi$ is again a $2k \times 2k$ Hermitian positive
matrix-valued measure on $\mathcal{F}$.

\par By (\ref{eq:2.3}) 
\begin{eqnarray}
L(a) -\wtilde{K}(a) &=& \frac{\imath}{2} \mathbbm{1}_{\{0\}} (a)
J_{2k} \nonumber \\
&=& \left[ \frac{\imath}{2} \int_{\what{D}}
\chi(a)\lambda(\dd\chi)\right]J_{2k}, \quad a \in D
\label{eq:2.11}
\end{eqnarray}
Subtracting (\ref{eq:2.9}) from (\ref{eq:2.7}) and using
(\ref{eq:2.11}) we have 
\[\int_{\what{D}} \chi(a)(\Psi -\Phi) (\dd\chi) = \left[
\frac{\imath}{2} \int_{\what{D}}
\chi(a)\lambda(\dd\chi)\right]J_{2k} \]
for all $a \in D$. Thus by uniqueness of Fourier
transform we have 
\[\Psi(S) - \Phi(S) = \frac{\imath}{2} \lambda(S)
J_{2k},\quad \forall S \in \mathcal{F}.\]
Thus 
\begin{equation} \label{eq:2.12}
\Phi(S) + \frac{\imath}{2} \lambda(S) J_{2k} \ge 0,\quad
\forall S \in \mathcal{F}.
\end{equation}
Now property (1) follows from (\ref{eq:2.10}), property (2)
from  (\ref{eq:2.9}), and property (3) from (\ref{eq:2.12}).
If $\tr \Phi(S) =0$ then $\Phi(S) =0$ and (\ref{eq:2.12})
implies 
\[\frac{\imath}{2} \lambda(S) J_{2k} \ge 0\]
and this is positive only if $\lambda(S) =0$. In other words
$\lambda \ll \tr \Phi$.

\par To prove property (4) of $\Phi$ we introduce the map
$\tau: \what{D} \rightarrow \what{D},\, \tau(\chi) =
\overline{\chi}= \chi^{-1}$ and observe that 
\begin{eqnarray*}
\wtilde{K}(a) = \overline{\wtilde{K}(a)} &=& \int_{\what{D}}
\overline{\chi(a)}\, \overline{\Phi}(\dd\chi)\\
&=& \int \chi(a)\, \overline{\Phi}\tau^{-1}(\dd\chi) \\
&=& \int \chi(a)\, \Phi(\dd\chi).
\end{eqnarray*} 
Thus 
\[\overline{\Phi} \tau^{-1} = \Phi,\]
or 
\[\Phi(S^{-1}) = \overline{\Phi(S)} = \Phi^T(S),\quad
\forall S \in \mathcal{F}.\]
This completes the proof of necessity. To prove sufficiency
consider a $2k \times 2k$ Hermitian positive matrix-valued
measure $\Phi$ satisfying properties (3) and (4) of the
theorem. Define 
\[\wtilde{K}(a) =\int_{\what{D}} \chi(a)\, \Phi(\dd\chi).\]
Property (3) implies that the function $L(a), \, a\in D$
defined by 
\begin{eqnarray*}
L(a) &=&\wtilde{K}(a) + \frac{\imath}{2} \mathbbm{1}_{\{0\}} (a)
J_{2k}\\
&=& \int_{\what{D}} \chi(a) \, \left( \Phi +\frac{\imath}{2}
\lambda J_{2k} \right) (\dd\chi)
\end{eqnarray*}
satisfies the matrix inequalities $[[ L(a_s - a_r) ]] \ge 0$
for any sequence $(a_1, \cdots, a_n) \in \mathcal{S}_D$. By
Proposition \ref{prop:2.1}, $\wtilde{K}$ is the autocovariance
function of a strictly stationary mean zero $k$-mode quantum
Gaussian process.
\end{proof}
\begin{rem}
As already described in the introduction, we call equation
(2) in Theorem \ref{th:2.2}, the \emph{spectral
representation} of the autocovariance map $\wtilde{K}$ and
say that $\Phi$ is the \emph{Kolmogorov-Wiener-Masani (KWM)
spectrum} of the $k$-mode weakly stationary quantum process.
Theorem \ref{th:2.2} enables us to construct a whole class
of examples of KWM spectra and hence autocovariance maps as
follows. Choose and fix any Borel map $\chi \mapsto M(\chi),
\, \chi \in \what{D}$ where $M(\chi)$ is a $k$-mode quantum
covariance matrix of order $2k$, so that
\[M(\chi) +\frac{\imath}{2} J_{2k} \ge 0, \quad \text{ for
every } \chi\in D.\]
Assume that $M(\cdot)$ is integrable with respect to the
normalised Haar measure $\lambda$ on $\what{D}$. Let $\Psi$
be any totally finite positive Hermitian $2k \times 2k$
matrix-valued measure on $(\what{D}, \mathcal{F})$
satisfying the conjugate symmetry condition $\Psi(S^{-1}) =
\overline{\Psi(S)}$ for any Borel set $S \subset \what{D}$.
Define
 \[\Phi(S) = \int_{\what{D}} M(\chi) \lambda(\dd\chi)
+\Psi(S),\quad S \in \mathcal{F}.\]
Then by Theorem \ref{th:2.2}, $\Phi$ is the KWM spectrum of a
stationary quantum Gaussian process over $D$ with
autocovariance map $\wtilde{K}$ given by equation
(2) of the theorem.

\end{rem}

\begin{rem}
\par The second part of property (3) of $\Phi$ in Theorem
\ref{th:2.2} implies that $\lambda(S)=0$ whenever $\tr
\Phi(S)=0$. In other words the KWM spectrum of a weakly
stationary $k$-mode quantum process over $D$ cannot admit a
gap of positive Haar measure in $\what{D}$. For example,
when $D=\mathbb{Z}$ and $\what{D}$ is identified with
$[0,2\pi]$, the KWM spectrum of a stationary  $k$-mode
quantum Gaussian process over $\mathbb{Z}$ cannot admit an
interval gap. 
\end{rem}

\begin{rem}\label{rem:4.3}
\par In Theorem \ref{th:2.2}, express the KWM spectrum
$\Phi$ as 
\[\Phi = [[\phi_{rs}]],\, r,s\in\{1,2,\cdots,2k\}\]
and write 
\begin{eqnarray*}
\Phi_q &=& [[ \phi_{2i-1,2j-1}]], \quad i,j \in
\{1,2,\cdots, k\}\\
\Phi_p &=& [[ \phi_{2i,2j}]], \quad i,j \in
\{1,2,\cdots, k\}.
\end{eqnarray*}
In the inductive limit C* probability space
$(\mathcal{B}_\infty, \omega)$ associated with the process
$\bm{\varrho}$ described in \S\ref{sec:2} the commuting family of position
observables $\{q_{a\,r}: a\in D, \, r \in \{1,2,\cdots,
k\}\}$ affiliated to $\mathcal{B}_\infty$ execute a
classical weakly stationary process with spectrum $\Phi_q$.
A similar property holds for the family $\{p_{a\,r}: a\in D,
\, r \in\{1,2,\cdots, k\}\}$.

\par This raises the question that under what conditions on the
spectrum $\Phi$, these processes enjoy properties like
ergodicity, weak mixing, strong mixing etc. The results of
G. Maruyama \cite{maru1} suggest that a minimum requirement
would be the absence of atoms in the spectrum $\Phi$. 
\par Suppose $\Phi$ has no atoms. For arbitrary real
scalars $c_r, \, 1 \le r \le 2k$, consider the associated
observables 
\[Z_a = \sum_{r=1}^{2k} \left( c_{2r-1} q_{ar} + c_{2r}
p_{ar} \right), \quad a \in D.\]
Then $\{ Z_a, a\in D\}$ executes a classical Gaussian
process with values in the real line and autocovariance
function $\bm{c}^T K(\cdot) \bm{c}$ with $\bm{c}^T= (c_1,
c_2, \cdots, c_{2k})$ and spectrum equal to $\bm{c}^T
\Phi(\cdot) \bm{c}$, a measure in $\what{D}$. Now let $D=
\mathbb{Z}$ be the integer group. Then Maruyama's theorem
implies that this scalar-valued process is, indeed, weakly
mixing, and in particular ergodic. If $\lim_{a \rightarrow
\infty} \bm{c}^T K(a) \bm{c} =0$, then this scalar-valued
process is also strongly mixing.
\end{rem}
\begin{rem}
\par Following \cite{krprb} one can introduce the observable 
\[N_{a\,j} = \frac{1}{2} \left( q_{a\,j}^2 + p_{a\,j}^2 - 1
\right), \quad 1 \le j \le k,\]
which is the number of particles (photons) in the $j$-th
mode at the site $a$. If the underlying process
$\bm{\varrho}$ is Gaussian with mean $\bm{0}$ then
\[ \braket{N_{a\,j}} = \frac{1}{2} \left\{ \phi_{2j-1, 2j-1}
(\what{D}) + \phi_{2j,2j}(\what{D}) - \frac{1}{2}
\right\}.\]
This is a consequence of property (1) of $\Phi$ in the
Theorem \ref{th:2.2} and Corollary 3.1, equation (3.4) \cite{krprb}.

\par This shows that the relationships between photon
numbers and KWM spectrum need a deeper exploration. 
\end{rem}

\par Theorem \ref{th:2.2} can be strengthened as follows:
\begin{theorem}\label{th:4.2}
Let $D$ be a countable, discrete, additive abelian group
with its compact character group $\what{D}$ and let
$\lambda$ be the normalized Haar measure of $\what{D}$ on
its Borel $\sigma$-algebra $\mathcal{F}$. Suppose $\Phi$ is
a complex, totally finite $2k \times 2k$ positive matrix-valued 
measure on $\mathcal{F}$. Then $\Phi$ is the KWM
spectrum of a $k$-mode weakly stationary quantum process
$\bm{\varrho}$ over $D$ if and only if $\Phi$ admits the
representation 
\begin{equation}\label{q1}
\Phi(S) = \int_S F(\chi) \lambda(\dd \chi) + \Psi(S), \quad
S \in \mathcal{F}
\end{equation}
where $F$ is a $2k \times 2k$ complex positive matrix-valued
Borel function satisfying the matrix inequality 
\begin{equation}\label{q2}
F(\chi) +\frac{\imath}{2} J_{2k} \ge 0, \quad \chi \in D,
\end{equation}
and 
\begin{equation}\label{q3}
\Psi(S) = [[ \psi_{rs}(S)]], \quad S \in \mathcal{F}
\end{equation}
is a $2k \times 2k$ positive matrix-valued measure on
$\mathcal{F}$ with each $\psi_{rs}$ being singular with
respect to $\lambda$.
\par In particular, $F$ can be chosen to satisfy 
\begin{equation}\label{q4}
\det \re F(\chi) \ge \frac{1}{4^k} ,\quad  \chi \in \what{D}
\end{equation}
and the absolutely continuous part of $\bm{u}^T \Phi(\cdot)
\bm{u}$ is equivalent to $\lambda$ for every nonzero element
$\bm{u}$ of $\mathbb{C}^{2k}$.
\end{theorem}
\begin{proof}
 \par Let $\Phi = [[\phi_{rs}]], \, r,s \in \{1,2,\cdots,
2k\}$ be the KWM spectrum of a $k$-mode weakly stationary
quantum process $\bm{\varrho}$ over $D$. By part (3) of
Theorem \ref{th:2.2} it follows that
\[ \begin{bmatrix}
\phi_{2r-1, 2r-1} (S) & \phi_{2r-1, 2r} (S)
+\frac{\imath}{2} \lambda(S) \\
\phi_{2r, 2r-1}(S) - \frac{\imath}{2} \lambda(S) & \phi_{2r, 2r}(S)
\end{bmatrix} \ge 0 \]
for any $S \in \mathcal{F}$. Suppose $\phi_{2r-1,
2r-1}(S)=0$. Then positivity of $\Phi$ implies that
$\phi_{2r-1, 2r}(S) = \phi_{2r, 2r-1}(S) =0$ and hence 
\[ \begin{bmatrix}
0 & \frac{\imath}{2} \lambda(S)\\
- \frac{\imath}{2} \lambda(S) & \phi_{2r, 2r}(S) 
\end{bmatrix} \ge 0.\]
Since the determinant of the left hand side in the inequality is
nonnegative we conclude that $\lambda(S) = 0$. In other
words $\lambda \ll \phi_{2r-1, 2r-1}$. By the same argument
$ \lambda \ll \phi_{2r, 2r}$. In other words $\lambda$ is
absolutely continuous with respect to every diagonal entry
of $\Phi$.

\par  Choose and fix an $S_0 \in \mathcal{F}$ such that $S_0
= S_0^{-1},\, \lambda(S_0) =1$ and the measure $\mu_{rr}$
defined by 
\[\mu_{rr} (S) = \phi_{rr} (S \cap S_0) , \quad S \in
\mathcal{F}\]
is the part of $\phi_{rr}$ equivalent to $\lambda$ and the
measure $\psi_{rr}$ defined by 
\[\psi_{rr}(S) = \phi_{rr} (S \cap (\what{D} \backslash S_0)), \quad S
\in \mathcal{F}\]
is the part of $\phi_{rr}$ singular with respect to
$\lambda$, so that 
\[ \phi_{rr} = \mu_{rr} + \psi_{rr}, \quad \forall r
=1,2,\cdots, 2k.\]
Now define
\begin{eqnarray*}
\mu_{rs}(S) &=& \phi_{rs} (S \cap S_0), \quad S \in
\mathcal{F},\\
\psi_{rs}(S) &=& \phi_{rs} (S \cap (\what{D} \backslash S_0)), \quad S
\in \mathcal{F}.
\end{eqnarray*}
If $\lambda(S) =0$, then $\mu_{rr}(S) =0$, so $\phi_{rr}(S
\cap S_0) =0, \, \phi_{rs}(S \cap S_0) =0$ and therefore
$\mu_{rs}(S) =0$. In other words $\mu_{rs}\ll \lambda$. By
definition all the measures  $\psi_{rs}$ are singular with
respect to $\lambda$ and 
\[ \phi_{rs} = \mu_{rs} + \psi_{rs}, \quad r,s \in
\{1,2,\cdots, 2k\}.\]
Define $f_{rs}$ to be the Radon Nykodym derivative of
$\mu_{rs}$ with respect to $\lambda$ and put
\[F(\chi) = f_{rs}(\chi),  \quad r,s \in
\{1,2,\cdots, 2k\}.\]
Now $F$ is defined a.e. with respect to $\lambda$ and
\begin{equation}\label{q5}
\Phi(S) = \int_S F(\chi) \lambda(\dd \chi) + \Psi(S), \quad
S \in \mathcal{F}
\end{equation}
where every entry $\psi_{rs}$ of $\Psi$ is singular with
respect to $\lambda$. By the choice of $S_0$ it follows that 
\[ \int_S F(\chi) \lambda(\dd \chi) + \frac{\imath}{2}
\lambda(S) J_{2k} \ge 0\]
for all $S \subset S_0, \, S \in \mathcal{F}$, so that 
\[ \int_S \left( F(\chi) + \frac{\imath}{2} J_{2k} \right)
\lambda(\dd\chi) \ge 0.\]
Thus 
\begin{equation}\label{q6}
F(\chi) + \frac{\imath}{2} J_{2k} \ge 0, \quad \text{ a.e.
} \chi(\lambda) 
\end{equation}
and also by the symmetry of the Haar measure
\begin{equation}\label{q7}
F(\chi^{-1}) + \frac{\imath}{2} J_{2k} \ge 0.
\end{equation}
On the other hand
\begin{eqnarray*}
\Phi(S^{-1}) &=& \int_S F(\chi^{-1}) \lambda(\dd\chi) +
\Psi(S^{-1}),\\
\overline{\Phi(S)} &=& \int_S \overline{F(\chi)}
\lambda(\dd\chi) + \overline{\Psi(S)}.
\end{eqnarray*}
The conjugate symmetry of $\Phi$ and the consequent
conjugate symmetry of $\Psi$, thanks to the choice of $S_0$
imply 
\[F(\chi^{-1}) = \overline{F(\chi)}, \quad \text{ a.e. }
\chi(\lambda). \]
Choosing $F(\chi) = \frac{1}{2} I_{2k}$ whenever this fails
we may assume that $ F(\chi^{-1}) = \overline{F(\chi)}$ for
all $\chi$. Now (\ref{q6}) and (\ref{q7}) imply that $F$ can
be altered on a set of $\lambda$-measure zero so that 
\[\re F(\chi) = \frac{F(\chi) + F(\chi^{-1})}{2} \ge
-\frac{\imath}{2} J_{2k}\]
holds for every $\chi$. In other words $F$ is a complex positive $2k \times 2k$
matrix whose real part is the quantum covariance matrix of
position and momentum observables in a $k$-mode state in
$L^2(\mathbb{R}^k)$. It follows from \cite{MR3070484} that
\[\det \re F(\chi) \ge \frac{1}{4^k}  \quad \forall \chi \in
\what{D}\]
and the representation (1) holds. 
\par The converse is already a part of Theorem \ref{th:2.2}.
\end{proof}

\begin{rem}
Suppose the quantum process $\bm{\varrho}$ of Theorem
\ref{th:4.2} is Gaussian and symmetric under the reflection
transformation $a \rightarrow -a$ in $D$. Then the
autocovariance function $K$ of the process $\bm{\varrho}$
satisfies the condition $K(a) = K(-a), \, a \in D$ and the
KWM spectrum $\Phi$ is real, i.e., $\Phi = \overline{\Phi}$.
Then (\ref{q4}) implies 
\[\int_{\what{D}} \log \det \Phi (\chi) \lambda(\dd\chi) > -
\infty . \]
When $D= \mathbb{Z}$ is the integer group it follows from
the Wiener Masani theorem, in particular, that the position
observables $(q_{n\,1}, \cdots, q_{n\,k})$ execute a purely
indeterministic shift invariant $k$-variate Gaussian
process. So do the momentum observables $(p_{n\,1}, \cdots,
p_{n\,k})$.
\end{rem}
\bibliographystyle{acm}
\bibliography{biblio}

\end{document}